\documentclass{llncs}

\usepackage{amsfonts}
\usepackage{amsmath}
\usepackage{latexsym}
\usepackage{color}
\usepackage{graphicx}
\usepackage{enumerate}
\usepackage{amstext}
\usepackage{epsfig}
\usepackage{pseudocode}
\usepackage[whileod]{algochl}
\usepackage{gastex}
\usepackage[mathscr]{euscript}

\def\Nset{\mathbb{N}}
\def\Rset{\mathbb{R}}
\def\Sset{\mathbb{S}}

\DeclareMathOperator{\Set}{set}
\DeclareMathOperator{\orig}{orig}
\DeclareMathOperator{\dest}{dest}

\DeclareMathOperator{\sset}{{\mathrm set}}

\newcommand{\e}{\epsilon}

\renewcommand{\l}{{\cal L}}
\newcommand{\f}{{\cal F}}
\newcommand{\set}[1]{\{#1\}}
\newcommand{\wI}{{w_{\mathcal I}}}
\newcommand{\WI}{{W_{\mathcal I}}}

\newcommand{\zero}{\overline{0}}
\newcommand{\one}{\overline{1}}

\newcommand{\Ssem}{(\Sset, \oplus, \otimes, \zero, \one)}

\newcommand{\Tsem}{(\Rset_+ \!\cup\! \{+\infty\}, \min, +, +\infty, 0)}

\newcommand{\R}{\mathsf R}

\newcommand{\Dis}{{\sc Disambiguation}}
\newcommand{\predis}{$\R$-pre-disambi\-guation}
\newcommand{\Predis}{$\R$-Pre-disambiguation}
\newcommand{\predisable}{$\R$-pre-disambiguable}
\newcommand{\predisability}{$\R$-pre-disambiguability}
\newcommand{\ipsfig}[2]{\scalebox{#1}{\psfig{#2}}}

\newcommand{\ignore}[1]{}

\title{On the Disambiguation of Weighted Automata}
\titlerunning{On the Disambiguation of Weighted Automata}

\author{Mehryar Mohri\inst{1,2}
\and Michael D. Riley\inst{2}}
\institute{Courant Institute of Mathematical Sciences, New York, NY. 
\and Google Research, New York, NY.}

\begin{document}
\maketitle

\begin{abstract}

  We present a disambiguation algorithm for weighted automata. The algorithm admits two main stages: a pre-disambiguation stage followed by a transition removal stage. We give a detailed description of the algorithm and the proof of its correctness. The algorithm is not applicable to all weighted automata but we prove sufficient conditions for its applicability in the case of the tropical semiring by introducing the \emph{weak twins property}. In particular, the algorithm can be used with all acyclic weighted automata, relevant to applications. While disambiguation can sometimes be achieved using determinization, our disambiguation algorithm in some cases can return a result that is exponentially smaller than any equivalent deterministic automaton. We also present some empirical evidence of the space benefits of disambiguation over determinization in speech recognition and machine translation applications.

\end{abstract}

\section{Introduction}

Weighted finite automata and transducers are widely used in
applications.  Most modern speech recognition systems used for
hand-held devices or spoken-dialog applications use weighted automata
and their corresponding algorithms for the representation of their
models and their efficient combination and search
\cite{MohriPereiraRiley2008,AllauzenBensonChelbaRileySchalkwyk2012}.
Similarly, weighted automata are commonly used for a variety of tasks
in machine translation \cite{iglesias2011} and other
natural language processing applications \cite{KaplanKay1994},
computational biology \cite{Durbin1998}, image processing
\cite{AlbertKari2009}, optical character recognition \cite{Bruel2008},
and many other areas.

A problem that arises in several applications is that of
\emph{disambiguation of weighted automata}: given an input weighted
automaton, the problem consists of computing an equivalent weighted
automaton that is \emph{unambiguous}, that is one with no two
accepting paths labeled with the same string. The need for
disambiguation is often motivated by the computation of the marginals
given a weighted transducer, or the common problem of determining the
most probable string or more generally the $n$ most likely strings, $n
\geq 1$, of a \emph{lattice}, an acyclic weighted automaton generated
by a complex model, such as those used in machine translation, speech
recognition, information extraction, and many other natural language
processing and computational biology systems. A lattice compactly
represents the model's most likely hypotheses. It defines a
probability distribution over the strings and is used as follows: the
weight of an accepting path is obtained by multiplying the weights of
its component transitions and the weight of a string obtained by
summing up the weights of accepting paths labeled with that string.
In general, there may be many accepting paths labeled with a given
string. Clearly, if the lattice were unambiguous, a standard
shortest-paths or $n$-shortest-paths algorithm \cite{Eppstein1998}
could be used to efficiently determine the $n$ most likely strings.
When the lattice is not unambiguous, the problem is more complex and
can be solved using weighted determinization \cite{MohriRiley2002}.
An alternative solution, which we will show has benefits, consists of
first finding an unambiguous weighted automaton equivalent to the
lattice and then running an $n$-shortest-paths algorithm on the
resulting weighted automaton.

In general, one way to determine an equivalent unambiguous weighted
automaton is to use the weighted determinization algorithm
\cite{Mohri1997}.  This, however, admits several drawbacks. First,
weighted determinization cannot be applied to all weighted
automata. This is both because not all weighted automata admit an
equivalent deterministic weighted automaton but also because even for
some that do, the weighted determinization algorithm may not
halt. Sufficient conditions for the application of the algorithm have
been given \cite{Mohri1997,AllauzenMohri2003}. In particular the
algorithm can be applied to all acyclic weighted
automata. Nevertheless, a second issue is that in some cases where
weighted determinization can be used, the size of the resulting
deterministic automaton is prohibitively large.

This paper presents a new disambiguation algorithm for weighted
automata. As we shall see, the algorithm applies to a broader family
of weighted automata. We show that, for the tropical semiring, if a
weighted automaton can be determinized using the algorithm of
\cite{Mohri1997}, then it can also be disambiguated using the
algorithm presented in this paper. Furthermore, for some weighted
automata, the size of the unambiguous weighted automaton returned by
our algorithm is exponentially smaller than that of any equivalent
deterministic weighted automata. In particular, our algorithm leaves
the input unchanged if it is unambiguous, while the size of the
automaton returned by determinization for some unambiguous weighted
automata is exponentially larger.  We also present empirical evidence
that shows the benefits of weighted disambiguation over
determinization in applications. Our algorithm applies in particular
to unweighted finite automata\ignore{, generalizing the disambiguation
algorithm of \cite{Mohri2012}}. Note that it is known that for some
non-deterministic finite automata of size $n$ the size of an
equivalent unambiguous automaton is at least $\Omega(2^{\sqrt{n}})$
\cite{Schmidt1978}, which gives a lower bound on the time and space
complexity of any disambiguation algorithm for finite automata.

We are not aware of any prior disambiguation algorithm for weighted
automata that is broadly applicable.\footnote{An algorithm
  of Eilenberg \cite{Eilenberg1974} bears the same name but it is in
  fact designed for an entirely different problem.} Nevertheless, our
algorithm is limited in some ways. First, not all weighted automata
admit an equivalent unambiguous weighted automaton. But, even for some
that do, our algorithm may not succeed. The situation is thus similar
to that of weighted determinization. However, we present sufficient
conditions under which our algorithm can be used, which covers all
acyclic weighted automata. Our algorithm has two stages. The first
stage called \emph{pre-disambiguation} constructs a weighted automaton
which has some key properties. In particular, the weight of
all paths leaving the initial state and labeled with the same string
is the same. The second stage consists of removing some transitions to
make the result unambiguous. Our disambiguation algorithm can be
applied whenever pre-disambiguation terminates.

The paper is organized as follows. In Section~\ref{sec:preliminaries},
we introduce some preliminary definitions and notation relevant to the
description of our algorithm.  Section~\ref{sec:predis} describes our
pre-disambiguation algorithm and proves some key properties of its
result. We describe in fact a family of pre-disambiguation algorithms
parameterized by a relation $R$ over the set of pairs of states. A
simple instance of that relation is for two states to be equivalent
when they admit a path labeled by the same string leading to a final
state.  In Section~\ref{sec:dis}, we describe the second stage, which
consists of transition removal, and prove the correctness of our
disambiguation algorithm.  In Section~\ref{sec:sufficient}, we
introduce the notion of \emph{weak twins property} which we use to
prove the sufficient conditions for the application of
pre-disambiguation and thus the full disambiguation algorithm. The
proofs for this section are given in the case of weighted automata
over the tropical semiring.  Finally, in
Section~\ref{sec:experiments}, we present experiments that compare
weighted disambiguation to determinization in speech recognition
and machine translation applications.  
Our implementation of these algorithms used in these
experiments is available through a freely available OpenFst library
\cite{AllauzenRileySchalkwykSkutMohri2007}. Detailed proofs for most
of our results are given in the appendix.

\section{Preliminaries}
\label{sec:preliminaries}

Given an alphabet $\Sigma$, we will denote by $|x|$ the length of a
string $x \in \Sigma^*$ and by $\e$ the \emph{empty string} for which
$|\e| = 0$.

The weighted automata we consider are defined over a broad class of
\emph{semirings}. A semiring is a system $\Ssem$ where $(\Sset,
\oplus, \zero)$ is a commutative monoid with $\zero$ as the identity
element for $\oplus$, $(\Sset, \otimes, \one)$ is a monoid with $\one$
as the identity element for $\otimes$, $\otimes$ distributes over
$\oplus$, and $\zero$ is an annihilator for $\otimes$.

A semiring is said to be \emph{commutative} when $\otimes$ is
commutative. Some familiar examples of (commutative) semirings are the
tropical semiring $\Tsem$ or the semiring of non-negative integers
$(\Nset, +, \times, 0, 1)$.  The multiplicative operation of a
semiring $\Ssem$ is said to be {\em cancellative} if for any $x$, $x'$
and $z$ in $\Sset$ with $z \neq \zero$, $x \otimes z = x' \otimes z$
implies $x = x'$. When that property holds, the semiring $\Ssem$ is
also said to be \emph{cancellative}.

A semiring $\Ssem$ is said to be {\em left divisible} if any element
$x \in \Sset -\set{\zero}$ admits a left inverse $x' \in \Sset$, that
is $x' \otimes x = \one$. $\Ssem$ is said to be {\em weakly left
  divisible} if for any $x$ and $x'$ in $\Sset$ such that $x \oplus x'
\neq \zero$, there exists at least one $z$ such that $x = (x \oplus
x') \otimes z$. When the $\otimes$ operation is cancellative, $z$ is
unique and we can then write: $z = (x \oplus x')^{-1} \otimes x$.

\ignore{
When $z$ is not unique, we can still assume that we have an algorithm
to find one of the possible $z$s and denote it by $(x \oplus x')^{-1}
\otimes x$. Furthermore, we will assume that $z$ can be found in a
consistent way, that is: $ ((u \otimes x) \oplus (u \otimes x'))^{-1}
\otimes (u \otimes x) = (x \oplus x')^{-1} \otimes x$ for any $x, x',
u \in \Sset$ such that $u \neq \zero$. A semiring is {\em
  zero-sum-free} if for any $x$ and $x'$ in $\Sset$, $x \oplus x' =
\zero$ implies $x = x' = \zero$.
}

Weighted finite automata (WFAs) are automata in which the transitions
are labeled with weights in addition to the usual alphabet symbols
which are elements of a semiring \cite{Kuich1986\ignore{,Mohri2009}}.  A WFA $A
= (\Sigma, Q, I, F, E, \lambda, \rho)$ over $\Sset$ is a $7$-tuple
where: $\Sigma$ is the finite alphabet of the automaton, $Q$ is a
finite set of states, $I \subseteq Q$ the set of initial states, $F
\subseteq Q$ the set of final states, $E$ a finite multiset of
transitions which are elements of $Q \times \Sigma \times \Sset \times
Q$, $\lambda\colon I \rightarrow \Sset$ an initial weight function,
and $\rho\colon F \rightarrow \Sset$ the final weight function mapping
$F$ to $\Sset$.

A path $\pi$ of a WFA is an element of $E^*$ with
consecutive transitions. We denote by $\orig[\pi]$ the origin state
and by $\dest[\pi]$ the destination state of the path.
A path is said to be \emph{accepting} or \emph{successful}
when $\orig[\pi] \in I$ and $\dest[\pi] \in F$.

We denote by $w[e]$ the weight of a transition $e$ and similarly by
$w[\pi]$ the weight of path $\pi = e_1 \cdots e_n$ obtained by
$\otimes$-multiplying the weights of its constituent transitions:
$w[\pi] = w[e_1] \otimes \cdots \otimes w[e_n]$. When $\orig[\pi]$ is
in $I$, we denote by $\wI[\pi] = \lambda(\orig[\pi]) \otimes w[\pi]$
the weight of the path including the initial weight of the origin
state.  For any two subsets $U, V \subseteq Q$ and any string $x \in
\Sigma^*$, we denote by $P(U, x, V)$ the set of paths labeled with $x$
from a state in $U$ to a state in $V$ and by $W(U, x, V)$ the
$\oplus$-sum of their weights:
\begin{equation*}
W(U, x, V) = \bigoplus_{\pi \in P(U, x, V)} w[\pi].\\[-.2cm]
\end{equation*}
When $U$ is reduced to a singleton, $U = \set{p}$, we will simply
write $W(p, x, V)$ instead of $W(\set{p}, x, V)$ and similarly for
$V$. To include initial weights, we denote:
\begin{equation*}
\WI(x, V) = \bigoplus_{\pi \in P(I, x, V)} \wI[\pi]. \\[-.2cm]
\end{equation*}
We also denote by $\delta(U, x)$ the set of states reached by
paths starting in $U$ and labeled with $x \in \Sigma^*$.
The weight
associated by $A$ to a string $x \in \Sigma^*$ is defined by
\begin{equation}
A(x) = \bigoplus_{\pi \in P(I, x, F)} \wI[\pi] \otimes \rho(\dest[\pi]),
\end{equation}
when $P(I, x, F) \neq \emptyset$. $A(x)$ is defined to be $\zero$ when $P(I, x, F) = \emptyset$. 

A state $q$ of a WFA $A$ is said to be {\em accessible} if
$q$ can be reached by a path originating in $I$. It is
\emph{coaccessible} if a final state can be reached by a path from
$q$.  A WFA $A$ is \emph{trim} if all states of $A$ are
both accessible and coaccessible. $A$ is \emph{unambiguous} if any
string $x \in \Sigma^*$ labels at most one accepting path.

In all that follows, we will consider weighted automata over
a weakly left divisible cancellative semiring.\footnote{The algorithms
  we present can be straightforwardly extended to the case of weakly
  left divisible left semirings \cite{AllauzenMohri2003}.}

\section{\Predis\ of weighted automata}
\label{sec:predis}

\subsection{Relation $\R$  over $Q \times Q$}

Two states $q, q' \in Q$ are said to share a common future if
there exists a string $x \in \Sigma^*$ such that 
$P(q, x, F)$ and $P(q', x, F)$ are
not empty. Let $\R^*$ be the relation defined over $Q \times Q$ by $q
\, \R^* \, q'$ iff $q$ and $q'$ share a common future in $A$.
Clearly, $\R^*$ is
reflexive and symmetric, but in general it is not transitive. Observe
that $\R^*$ is \emph{compatible with the inverse transition function},
that is, if $q \, \R^* \, q'$, $q \in \delta(p, x)$ and $q' \in
\delta(p', x)$ for some $x \in \Sigma^*$ with $(p, p') \in Q^2$, then
$p \, \R^* \, p'$.  We will also denote by $\R_0$ the complete
relation defined by $q \, \R_0 \, q'$ for all $(q, q') \in
Q^2$. Clearly, $R_0$ is also compatible with the inverse transition
function.

The construction we will define holds for any relation $\R$\ignore{
  over $Q \times Q$} out of the set of admissible relations ${\cal R}$
defined as the relations over $Q \times Q$ that are compatible with
the inverse transition function and coarser than $\R^*$. Thus, ${\cal
  R}$ includes $\R^*$ and $\R_0$, as well as any relation $\R$
compatible with the inverse transition function that is coarser than
$\R^*$, that is, for all $(q, q') \in Q^2$, $q \, \R^* \, q' \implies
q \, \R \, q'$. Thus, for a relation $\R$ in ${\cal R}$, two states
$q$ and $q'$ that share the same future are necessarily in relation,
but they may also be in relation without sharing the same future. Note
in particular that $\R$ is always reflexive\ignore{ and symmetric}.

\subsection{Construction}

Fix a relation $\R \in {\cal R}$. For any $x \in \Sigma^*$, and $q \in
\delta(U, x)$, we also denote by $\delta_q(U, x)$ the set of states in
$\delta(U, x)$ that are in relation with $q$:
\begin{equation*}
\delta_q(U, x) =
\delta(U, x) \cap \set{p\colon p \, \R \, q}.
\end{equation*}
Note that, since $\R$ is reflexive, by definition, $\delta_q(I, x)$
contains $q$.  
For any $x \in \Sigma^*$ and $q \in \delta(I, x)$, we define the
weighted subset $s(x, q)$ by
\begin{align*}
s(x, q) = 
& \Big\{(p_1, w_1), \ldots, (p_t, w_t)\colon \big( \set{p_1, \ldots,
  p_t} = \delta_q(I, x) \big)\\
& \ \wedge \big( \forall i \in [1, t], w_i = \WI(x, \set{p_1, \ldots, p_t})^{-1} \otimes \WI(x,
p_i)\big) \Big\}.
\end{align*}
For a weighted subset $s$, define $\Set(s) = \set{p_1, \ldots, p_t}$.
For any automaton $A$ define $A' = (\Sigma, Q', I', F', E', \lambda', \rho')$ as
follows:
\begin{align*}
& Q' = \set{(q, s(x, q)) \colon x \in \Sigma^*, q \in \delta(I, x)}\\
& I' = \set{(q, s(\e, q)) \colon q \in I} \quad \text{and} \quad 
F' = \set{(q, s(x, q)) \colon x \in \Sigma^*, q \in \delta(I, x)
  \cap F}\\
& E' = \bigg\{ 
((q, s), a, w, (q', s')) \colon 
 (q, s), (q', s') \in Q', a \in \Sigma,\\[-.15cm]
& \mspace{60mu} \exists x \in \Sigma^*\mid & & \mspace{-460mu} s = s(x, q) = \set{(p_1, w_1), \ldots, (p_t, w_t)}, \\
& & & \mspace{-460mu} s'= s(xa, q') = \set{(p'_1, w'_1), \ldots, (p'_{t'}, w'_{t'})},\\[-.15cm]
& & & \mspace{-460mu}q' \in \delta(q, a), w = \bigoplus_{i = 1}^t \Big (w_i \otimes W(p_i, a, \Set(s')) \Big),\\[-.4cm]
& & & \mspace{-460mu} \forall j \in [1, t'], w'_j = w^{-1} \otimes \Big( \bigoplus_{i = 1}^{t} w_i \otimes W(p_i, a,
p'_j \Big)
\bigg\}\\
\end{align*}
\begin{flalign*}
\text{and} \quad
& \forall (q, s) \in I', s = \set{(p_1, w_1), \ldots, (p_t, w_t)}, \,
\lambda'((q, s)) = \bigoplus_{\substack{i \in [1, t]}} \lambda(p_i).\\
& \forall (q, s) \in F', s = \set{(p_1, w_1), \ldots, (p_t, w_t)}, \,
\rho'((q, s)) = \bigoplus_{\substack{p_i \in F\\ i \in [1, t]}}
(w_i \otimes \rho(p_i)).\\[-.8cm]
\end{flalign*}
Note that in definition of the transition set $E'$ above, the property
$\Set(s') = \delta_{q'}(\Set(s), a)$ always holds.  In particular, if
$p'$ is in $\delta_{q'}(\Set(s), a)$, then there is a path from $I$ to
some $p \in \Set(s)$ labeled $x$ and a transition from $p$ to $p'$
labeled with $a$ and $p'\, R \, q'$ so $p'$ is in $\Set(s')$.
Conversely, if $p'$ is in $\Set(s')$ then there exists $p$ reachable
by $x$ with a transition labeled with $a$ from $p$ to $p'$. Since $p'$
is in $\Set(s')$, $p'$ is in $\delta_{q'}(I, xa)$, thus $p' \, R \,
q'$. Since there exists a transition labeled with $a$ from $q$ to $q'$
and from $p$ to $p'$, this implies that $p \, R \, q$.  Since $p \, R
\, q$ and $p$ is reachable via $x$, $p$ is $\delta_q(I, x)$.

When the set of states $Q'$ is finite, $A'$ is a WFA
with a finite set of states and transitions and is defined as the
result of the \emph{\predis\ of $A$}. In general,
\predis\ is thus defined only for a subset of weighted
automata, which we will refer to as the set of
\emph{\predisable\ weighted automata}. We will show
later sufficient conditions for an automaton $A$ to be
\emph{\predisable} in the case of the tropical
semiring. Figure~\ref{fig:1}
illustrates the \predis\ construction.

\subsection{Properties of the resulting WFA}

In this section, we assume that the input WFA $A = (\Sigma, Q,
I, F, E, \lambda, \rho)$ is \predisable. In general, the
WFA $A'$ constructed by \predis\ is
not equivalent to $A$, but the weight of each path from an initial
state equals the $\oplus$-sum of the weights of all paths with the
same label in the input automaton starting at an initial state.

\begin{proposition}
\label{prop:1}
Let $A' = (\Sigma, Q', I', F', E', \lambda', \rho')$ be the finite
automaton returned by the \predis\ of the WFA $A = (\Sigma, Q, I, F,
E, \lambda, \rho)$. Then, the following equalities hold for any path
$\pi \in P(I', x, (q, s))$ in $A'$, with $x \in \Sigma^*$ and $s =
\set{(p_1, w_1), \ldots, (p_t, w_t)}$:
\begin{align*}
\wI[\pi] = \WI(x, \Set(s)) \quad \text{ and } \quad
\forall i \in [1, t], \ \wI[\pi] 
 \otimes w_i & = \WI(x, p_i).
\end{align*}

\end{proposition}

\begin{proposition}
\label{prop:2}
Let $A' = (\Sigma, Q', I', F', E', \lambda', \rho')$ 
be the finite automaton returned by
the \predis\ of the WFA $A = (\Sigma, Q, I,
F, E, \lambda, \rho)$.  Then, for any accepting 
path $\pi \in P(I', x, (q, s))$ in
$A'$, with $x \in \Sigma^*$ and 
$(q, s) \in F'$, the following equality holds:
\begin{equation*}
\wI[\pi] \otimes \rho'((q, s)) = A(x).
\end{equation*}

\end{proposition}
\begin{proof}
Let $s = \set{(p_1, w_1), \ldots, (p_t,
  w_t)}$.
By definition of $\rho'$, we can write
\begin{align*}
\wI[\pi] \otimes \rho'((q, s)) 
& = \wI[\pi] \otimes \bigoplus_{\substack{p_i \in F\\ i \in [1, t]}}
(w_i \otimes \rho(p_i))
= \bigoplus_{\substack{p_i \in F\\ i \in [1, t]}}
(\wI[\pi] \otimes w_i \otimes \rho(p_i)).\\[-.25cm]
\end{align*}
Plugging in the expression of $(\wI[\pi] \otimes w_i)$ given by
Proposition~\ref{prop:1} yields
\begin{equation}
\label{eq:4}
\wI[\pi] \otimes \rho'((q, s)) 
= \bigoplus_{\substack{p_i \in F\\ i \in [1, t]}}
( \WI(x, p_i) \otimes \rho(p_i)). \\[-.25cm]
\end{equation}
By definition of \predis, $q$ is a final state. Any state $p \in \delta(I, x) \cap F$ shares a common future 
with $q$ since both $p$ and $q$ are final states, thus we must have $p \,
R \, q$, which implies
$p \in \Set(s)$. Thus, the $\oplus$-sum in \eqref{eq:4} is exactly
over the set of states $\delta(I, x) \cap F$, which proves that
$\wI[\pi] \otimes \rho'((q, s)) = A(x)$.\qed
\end{proof}

\begin{proposition}
\label{prop:3}
Let $A' = (\Sigma, Q', I', F', E', \lambda' \rho')$ be the finite 
automaton returned by
the \predis\ of the WFA $A = (\Sigma, Q, I,
F, E, \lambda, \rho)$.  Then, any string $x \in \Sigma^*$ accepted by $A$ is
accepted by $A'$.
\end{proposition}
\begin{proof}
  Let $(q_0, a_1,
  q_1) \cdots (q_{n - 1}, a_n, q_n)$ be an accepting path in $A$ with
  $a_1 \cdots a_n = x$.  By construction, 
$((q_0, s_0), a_1,
  (s_1, q_1)) \cdots ((s_{n - 1}, q_{n - 1}), a_n, (s_n, q_n))$
is a path in $A'$ with $s_i = s(a_1 \cdots a_i, q_i)$ for
all $i \in [1, n]$
and $s_0 = \e$ and by definition of finality in \predis, 
$(s_n, q_n)$ is final. Thus, $x$ is accepted by $A'$. \qed
\end{proof}

\begin{figure}[t]
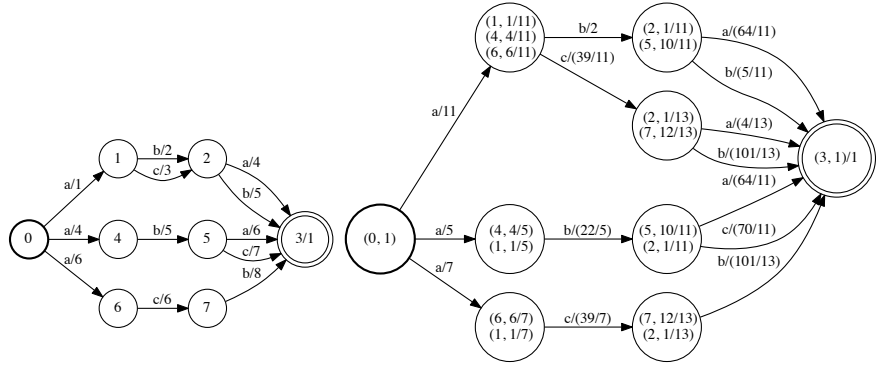

\centering
\begin{tabular}{cc}
\raisebox{.5cm}{\ipsfig{.4}{figure=t1}} & \ipsfig{.4}{figure=t3} 
\end{tabular}
\caption{Illustration of the \predis\ construction in the semiring
  $(\Rset_+, +, \times, 0, 1)$. For each state $(q, s)$ of the result,
  the subset $s$ is explicitly shown. $q$ is the state of the first
  pair in $s$ shown. The weights are rational numbers, for example
  $1/11 = \frac{1}{11} \approx .091$.}
\label{fig:1}
\end{figure}

\section{Disambiguation algorithm}
\label{sec:dis}

Propositions~\ref{prop:1}-\ref{prop:3} show that the strings accepted
by $A'$ are exactly those accepted by $A$ and that the weight of any
path in $A'$ accepting $x \in \Sigma^*$ is $A(x)$. Thus, if for any
$x$, we could eliminate from $A'$ all but one of the paths labeled
with $x$, the resulting WFA would be unambiguous and equivalent to
$A$. Removing transitions to achieve this objective without changing
the function represented by the WFA turns out not to be
straightforward. The following two lemmas (Lemmas \ref{lemma:1} and
\ref{lemma:11}) and their proofs are the critical technical
ingredients helping us define the transition removal and prove its
correctness. This first lemma provides a useful tool for the proof of
the second.

\begin{lemma}
\label{lemma:1}
  Let $A' = (\Sigma, Q', I', F', E', \lambda', \rho')$ be 
  the finite automaton returned
  by the \predis\ of the WFA $A = (\Sigma, Q, I, F, E,
  \lambda, \rho)$. Let $(q, s)$ and $(q', s')$ be two distinct states of $A'$
  both admitting a transition labeled with $a \in \Sigma$ to the same
  state $(q_0, s_0)$ (or both final states), and such that $(q, s) \in
  \delta(I', x)$ and $(q', s') \in \delta(I', x)$ for some $x \in
  \Sigma^*$. Then, if $(q, s) \in \delta(I', x')$ for some $x' \neq
  x$, $x' \in \Sigma^*$, there exists a state $(q', s'') \in
  \delta(I', x')$ with $(q', s'') \neq (q, s)$ and such that $(q',
  s'')$ admits a transition labeled with $a$ to $(q_0, s_0)$ (resp. is
  a final state).
\end{lemma}

Let $A' = (\Sigma, Q', I', F', E', \lambda', \rho')$ be the 
finite automaton returned by
the \predis\ of the WFA $A = (\Sigma, Q, I, F, E, 
\lambda, \rho)$.  For
any state $(q_0, s_0)$ of $A'$ and label $a \in \Sigma$, let $\l(q_0,
s_0, a) = ((q_1, s_1), \ldots, (q_n, s_n))$, $n \geq 1$, be the list
of all distinct states of $A'$ admitting a transition labeled with $a
\in \Sigma$ to $(q_0, s_0)$, with $q_1 \leq \cdots \leq q_n$.  We
define the \emph{processing} of the list $\l(q_0, s_0, a)$ as follows:
the states of the list are processed in order; for each state $(q_j,
s_j)$, $j \geq 1$, this consists of removing its $a$-transition to
$(q_0, s_0)$ if and only if there exists a co-reachable state $(q_i,
s_i)$ with $i < j$ whose $a$-transition to $(q_0, s_0)$ has not been
removed.\footnote{This condition can in fact be relaxed: it suffices
  that there exists a co-reachable state $(q_i, s_i)$ with $i <
  j$ since it can be shown that in that case, there exists necessarily
  such a state with a $a$-transition to $(q_0, s_0)$. } Note that, by
definition, the $a$-transition to $(q_0, s_0)$ of the first state
$(q_1, s_1)$ is kept.

We define in a similar way the processing of the list $\f = ((q_1,
s_1), \ldots, (q_n, s_n))$, $n \geq 1$, of all distinct final states
of $A'$, with $q_1 \leq \cdots \leq q_n$ as follows: the states of the
list are processed in order; for each state $(q_j, s_j)$, $j \geq 1$,
this consists of making it non-final if and only if there exists a
co-reachable state $(q_i, s_i)$ with $i < j$ whose finality has been
maintained. By definition, the finality of state $(q_1, s_1)$ is
maintained.

\begin{lemma}
\label{lemma:11}
Let $A' = (\Sigma, Q', I', F', E', \lambda', \rho')$ be 
the finite automaton returned by
the \predis\ of the WFA 
$A = (\Sigma, Q, I, F, E, \lambda, \rho)$.  Let
$(q_0, s_0)$ be a state of $A'$ and $a \in \Sigma$, then, the
automaton $A''$ resulting from processing the list $\l(q_0, s_0, a)$
accepts the same strings as $A'$. Similarly, the processing of the
list of final states $\f$ of $A'$ does not affect the set of strings
accepted by $A'$.
\end{lemma}

Assume that $A$ is \predisable. Then, this helps us define a 
disambiguation algorithm \Dis\ for $A$ defined as follows:
\begin{enumerate}

\item construct $A'$, the result of the \predis\ of $A$;

\item for any state $(q_0, s_0)$ of $A'$ and label $a \in \Sigma$, 
  process $\l(q_0, s_0, a)$; process the list of final states $\f$.

\end{enumerate}

\begin{theorem}
  Let $A = (\Sigma, Q, I, F, E, \lambda, \rho)$ be a \predisable\ weighted
  automaton. Then, algorithm \Dis\ run on input $A$ generates
  an unambiguous WFA $B$ equivalent to $A$.
\end{theorem}
\begin{proof}
  Let $A' = (\Sigma, Q', I', F', E', \lambda', \rho')$ be the WFA
  returned by \predis\ run with input $A$. By lemma~\ref{lemma:11},
  the set of strings accepted after processing the lists $\l(q_0, s_0,
  a)$ and $\f$ remains the same\footnote{
The lemma is stated as processing one list, but from 
the proof it is clear it applies to multiple lists.}.
Furthermore, in view of the
  Propositions~\ref{prop:1}-\ref{prop:3}, the weight of the unique
  path labeled with an accepted string $x$ in $B$ $\otimes$-multiplied
  by its final weight is exactly $A(x)$. Finally, by definition of the
  processing operations, the resulting WFA is
  unambiguous, thus $B$ is an unambiguous WFA
  equivalent to $A$. \qed
\end{proof}
Differing numberings of the states can lead to different orderings in
each list and thus to different transition or finality removals,
thereby resulting in different weighted automata, with potentially
different sizes after trimming. Nevertheless, all such resulting
weighted automata are equivalent.

\section{Sufficient conditions}
\label{sec:sufficient}

The definition of siblings and that of twins property for weighted
automata were previously given by
\cite{choffrut1978,Mohri1997,AllauzenMohri2003}. We will use a weaker
(sufficient) condition for \predisability.

\begin{definition}
Two states $p$ and $q$ of a WFA $A$ are said to be
\emph{siblings} if there exist two strings $x, y \in \Sigma^*$ such
that both $p$ and $q$ can be reached from an initial state by paths
labeled with $x$ and there are cycles at both $p$ and $q$ labeled with
$y$.

Two sibling states $p$ and $q$ are said to be \emph{twins} if for any
such $x$ and $y$, $W(p, y, p) = W(q, y, q)$.  $A$ is said to have the
\emph{twins property} when any two siblings are twins. It is said to
have the $\R$-\emph{weak twins property} when any two siblings that
are in $\R$ relation are twins. When $A$ admits the $\R^*$-\emph{weak
  twins property}, we will also say in short that it admits the
\emph{weak twins property}.
\end{definition}

The results given in the remainder of this section are presented in
the specific case of the tropical semiring. To show the following
theorem we partly use a proof technique from \cite{Mohri1997} for showing
that the twins property is a sufficient condition for weighted
determinizability.

\begin{theorem}
\label{th:weak_twins}
  Let $A$ be a WFA over the tropical semiring that admits
  the $\R$-\emph{weak twins property}. Then, $A$ is \predisable.
\end{theorem}

The theorem implies in particular that if $A$ has the twins property
then $A$ is \predisable.  In particular, any acyclic weighted
automaton is \predisable.

A WFA $A$ is said to be \emph{determinizable} when the
weighted determinization algorithm of \cite{Mohri1997} terminates with input
$A$ (see also \cite{AllauzenMohri2003}). In that case, the output of the algorithm
is a deterministic automaton equivalent to $A$.

\begin{theorem}
\label{th:det_predis}
  Let $A$ be a determinizable WFA over the tropical
  semiring, then $A$ is \predisable.
\end{theorem}

By the results of \cite{kirsten08}, this also implies that any
polynomially ambiguous WFA that has the \emph{clones property} is
\predisable\ and can be disambiguated using \Dis. There are however
weighted automata that are \predisable\ and thus can be disambiguated
using \Dis but that cannot be determinized using the algorithm of
\cite{Mohri1997}.  Figure~\ref{fig:9} gives an example of such a WFA.
To see that the WFA $A$ of Figure~\ref{fig:9} cannot be determinized,
consider instead $B$ obtained from $A$ by removing the transition from
state $3$ to $5$. $B$ is unambiguous and does not admit the twins
property (cycles at states $1$ and $2$ have distinct weights), thus it
is not determinizable by theorem 12 of \cite{Mohri1997}. Weighted determinization
creates infinitely many subsets of the form $\set{(1, 0), (2, n)}$, $n
\in \Nset$, for paths from the initial state labeled with
$ab^n$. Precisely the same subets are created when applying
determinization to $A$.

\begin{figure}[t]
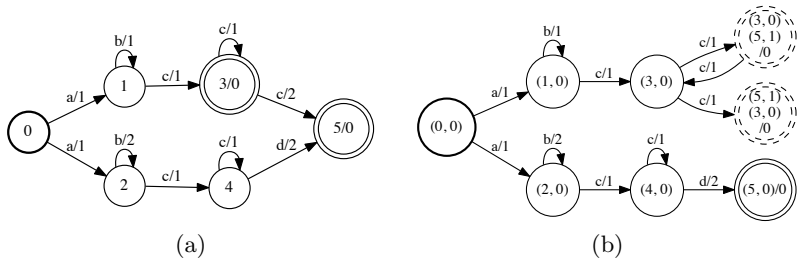

\centering
\begin{tabular}{c@{\hspace{.5cm}}c}
\raisebox{.15cm}{\ipsfig{.43}{figure=ex9}} & 
\ipsfig{.4}{figure=ex9_dis}\\
(a) & (b)
\end{tabular}
\caption{(a) Weighted automaton $A$ that cannot be determinized by the
  weighted determinization algorithm of \cite{Mohri1997}. (b) $A$ has
  the weak twins property and can be disambiguated by \Dis as shown by
  the figure.  One of the two states in dashed style is not made final
  by the algorithm. The head state for each of these states, is the
  state appearing in the first pair listed.}
\label{fig:9}
\end{figure}

The following result can be proven in a way that is similar to the
proof of the analogous result for the twins property given by
\cite{AllauzenMohri2003}.\footnote{In \cite{AllauzenMohri2003}, the
  authors use instead the terminology of \emph{cycle-unambiguous}
  weighted automata, which coincides with that of polynomially ambiguous
  weighted automata\ignore{ \cite{AllauzenMohriRastogi2011}}.}

\begin{theorem}
\label{th:4}
Let $A$ be a trim polynomially ambiguous WFA over the
tropical semiring. Then, $A$ has the weak twins property iff the
weight of any cycle in $B = \text{\sc Trim}(A \cap (-A))$ is $0$.
\end{theorem}

This leads to an algorithm for testing the weak twins property for
polynomially ambiguous automata in time $O(|Q|^2 + |E|^2)$. It was
recently shown that the twins property is a decidable property that is
PSPACE-complete for WFAs over the tropical semiring
\cite{Kirsten2012}. It would be interesting to determine if the weak
twins property we just introduced is also decidable.

\section{Experiments}
\label{sec:experiments}

In order to experiment with weighted disambiguation, we implemented
the algorithm in the {\em OpenFst} C++ library
\cite{AllauzenRileySchalkwykSkutMohri2007}. For comparison, an
implementation of weighted determinization is also available in that
library \cite{Mohri1997}.

For a first test corpus, we generated 500 speech {\em lattices} drawn from a
randomized, anonymized utterance sampling of voice searches on the
Google Android platform
\cite{SchalkwykBeefermanBeaufaysByrneChelbaCohenKamvarStrope2010}.
Each lattice is a weighted acyclic automaton over spoken words that
contains many weighted paths. Each path represents a hypothesis of
what was uttered along with the automatic speech recognizer's (ASR)
estimate of the probability of that path.  Such lattices are useful
for passing compact hypothesis sets to subsequent processing without
commitment to, say, just one solution at the current stage.

The size of a lattice is determined by a probability threshold with
respect to the most likely estimated path in the lattice; hypotheses
within the threshold are retained in the lattice.  Using $|A| = |Q| +
|E|$ to measure automata size, the mean size for these lattices was
2384 and the standard deviation was 3241.

The ASR lattices are typically non-deterministic and ambiguous due to
both the models and the decoding strategies used. Determinization can
be applied to reduce redundant computation in subsequent stages;
disambiguation can be applied to determine the combined probability
estimate of a string that may be distributed among several otherwise
identically-labels paths.

\ignore{
Figure~\ref{fig:disdet} shows the expansion in automata size (size of
result divided by size of input) with both disambiguation and
determinization for the 500 sample lattices.  }Disambiguation has a
mean expansion of 1.23 and a standard deviation of 0.59.
Determinization has a mean expansion of 1.31 and a standard deviation
of 1.35.  For this data, disambiguation has a slightly less mean
expansion compared to determinization but a very substantially less
standard deviation.

\ignore{
\begin{center}
\begin{figure}[t]
\begin{tabular}{cc}
\includegraphics[scale=0.33]{figs/disamb} &
\includegraphics[scale=0.33]{figs/det}\\
(a) & (b)\\[-.2cm]
\end{tabular}
\caption{Expansion in automata size with weighted (a) disambiguation
  and (b) weighted determinization; $|A| = |Q| + |E|$. Data are
  $500$ speech {\em lattices} from a voice search task.}
\label{fig:disdet}
\end{figure}
\vskip -1.0cm
\end{center}
}

As a second test corpus, we used 100 automata that are the compact
representation of hypothesized Chinese-to-English translations from
the DARPA Gale task \cite{iglesias2011}. These automata may contain
cycles due to details of the particular translation system,
which provides an interesting contrast to the acyclic speech case.
Some fail to determinize within the allotted memory (1GB)
and about two-thirds of those also fail to disambiguate, possible
when cycles are present.

Considering only those which are both determinizable and disambiguable,
disambiguation has a
mean expansion of 4.53 and a standard deviation of 6.0.
Determinization has a mean expansion of 54.5 and a standard deviation
of 90.5.  For this data, disambiguation has a much smaller mean
and standard deviation of expansion compared to determinization.

As a final example, Figure~\ref{fig:t4} shows an acyclic unambiguous
(unweighted) automaton whose size is in $O(n^2)$. No equivalent
deterministic automaton can have less than $2^n$ states since such an
automaton must have a distinct state for each of the prefixes of the
strings $\set{(a + b)^{k - 1} b (a + b)^{n - k} \colon 1 \leq k \leq
  n}$, which are prefixes of $L$. Thus, while our disambiguation
algorithm leaves the automaton of Figure~\ref{fig:t4} unchanged,
determinization would result in this case in an automaton with more
than $2^n$ states.

\begin{figure}[t]
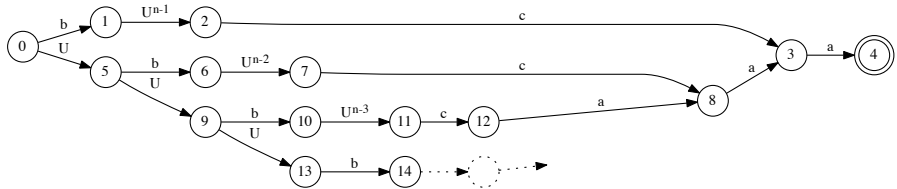

\centering
\vspace{-1cm}
\ipsfig{.4}{figure=t4} 
\vspace{-1cm}
\caption{Unambiguous automaton over the alphabet $\set{a, b, c}$
  accepting the language $L = \set{(a + b)^{k - 1} b (a + b)^{n - k} c
    a^k \colon 1 \leq k \leq n}$. For any $k \geq 0$, $U^k$ serves as
  a shorthand for $(a + b)^k$.}
\label{fig:t4}
\end{figure}

\section{Conclusion}

We presented the first algorithm for the disambiguation of WFAs.  The
algorithm applies to a family of WFAs verifying a sufficient condition
that we describe, which includes all acyclic WFAs. Our experiments
showed the favorable properties of this algorithm in applications
related to speech recognition and machine translation. The algorithm is likely to admit a
large number of applications in areas such as natural language
processing, speech processing, computational biology, and many other
areas where WFAs are commonly used. The study of the theoretical
properties we initiated raises a number of novel theoretical questions
which include the following: the decidability of the weak twins
property for arbitrary WFAs, the characterization of WFAs that admit
an equivalent unambiguous WFA, the characterization of WFAs to which
our algorithm can apply and perhaps an extension of our algorithm to a
wider domain, and finally the proof and study of these questions for
other semirings than the tropical semiring.

\ignore{

\section*{Acknowledgments}

We thank Cyril Allauzen for discussions.
This work was partly funded by the NSF award IIS-1117591.

}

\bibliographystyle{abbrv}
\bibliography{wdis}

\newpage
\appendix

\section*{Appendix}
\vskip -.38cm
\section{Proof of Proposition \ref{prop:1}}

\begin{proof}
  The proof is by induction on the length of $\pi$. If $\pi$ has
  length zero, it is a zero-length path from the state $(q,
  s(\e, q)) \in I'$ to the same state and
  $\wI[\pi] = \lambda'((q, s))$.
  We have $\WI(\e, \Set(s)) = \bigoplus_{\pi \in P(I, \e, \Set(s))} \wI[\pi] =
  \bigoplus_{p \in \Set(s)} \lambda(p) = \lambda'((q, s))$. Also, for
  all $i \in [1, t]$,
  $\wI[\pi] \otimes w_i = \lambda'((q,s)) \otimes \left [
  \lambda'((q, s))^{-1} \otimes \lambda(p_i) \right ] = \lambda(p_i)$ and
  $\WI(\e, p_i) = \lambda(p_i)$, thus the equalities trivially
  hold.

  Assume that the equalities hold for all paths of length at most $n
  \in \Nset$ starting in $I'$ and let $\pi$ be a path of length $n +
  1$.  We can therefore decompose $\pi$ as a path in $P(I, x, (q',
  s'))$ for some $x \in \Sigma^*$, $q \in Q$, and $s'= \set{(p'_1,
    w'_1), \ldots, (p'_{t'}, w'_{t'})} \in Q'$, followed by a transition
  $e = ((q', s'), a, w[e], (q, s))$ from $(q', s')$ to $(q, s)$. By
  definition of $w[e]$ in \predis\, we
can write
\begin{align}
\wI[\pi] 
 = \wI[\pi']  \otimes w[e] \nonumber
& = \wI[\pi'] \otimes \bigoplus_{j =
    1}^{t'} \Big (w'_j \otimes W(p'_j, a, \Set(s)) \Big) \nonumber\\[-.4cm]
& = \bigoplus_{j =
    1}^{t'} \wI[\pi'] \otimes  w'_j \otimes W(p'_j, a, \Set(s)) \nonumber\\[-.4cm]
\label{eq:1}
& = \bigoplus_{j =
    1}^{t'} \WI(x, p'_j) \otimes W(p'_j, a, \Set(s)),
\end{align}
where we used the identities $\wI[\pi'] \otimes w'_j = \WI(x, p'_j)$,
$j \in [1, t']$, which hold by the induction hypothesis.

We will show that any path $\xi$ in $A$ labeled with $xa$, starting in
$I$ and ending in $\Set(s)$ must go through $\Set(s')$, that is, $\xi$
can be decomposed into a path labeled with $x$ and reaching a state of
$\Set(s')$ followed by a transition labeled with $a$ from that state
to a state of $\Set(s)$. \eqref{eq:1} then implies that
\begin{equation}
\label{eq:12}
\wI[\pi] = \WI(xa, \Set(s)).
\end{equation}
Indeed, let $\xi = \xi' e'$ be a decomposition of $\xi$ into a path
$\xi'$ labeled with $x$ from $I$ to some state $p' \in Q$ followed by
a transition $e'$ labeled with $a$ from $p'$ to some state $p \in
\Set(s)$. By definition of \predis\, since $p$ is in $\Set(s)$, we
have $p \, \R \, q$.  By the compatibility of $\R$ with the inverse
transition function, $p \in \delta(p', a)$, and $q \in \delta(q', a)$,
this implies $p' \, \R \, q'$. Since we also have $p' \in \delta(I,
x)$, this shows that $p'$ is in $\Set(s')$ and therefore that path
$\xi'$ ends in $\Set(s')$.

In view of $\wI[\pi] = \wI[\pi'] \otimes w[e]$ and using the
definition of $w_i$ in \predis\, we can write, for any $i \in [1, t]$,
\begin{align}
\wI[\pi] \otimes w_i 
& = \wI[\pi'] \otimes w[e] \otimes w[e]^{-1} \otimes \Big( \bigoplus_{j = 1}^{t'} w'_j \otimes W(p'_j, a,
p_i) \Big) \nonumber\\[-.4cm]
& = \wI[\pi'] \otimes \Big( \bigoplus_{j = 1}^{t'} w'_j \otimes W(p'_j, a,
p_i) \Big)  \nonumber\\[-.4cm]
& = \bigoplus_{j = 1}^{t'} \wI[\pi'] \otimes w'_j \otimes W(p'_j, a,
p_i) \nonumber\\[-.4cm]
\label{eq:3}
& = \bigoplus_{j = 1}^{t'} \WI(x, p'_j) \otimes W(p'_j, a,
p_i),
\end{align}
using the identities $\wI[\pi'] \otimes w'_j = \WI(x, p'_j)$,
$j \in [1, t']$, which hold by the induction hypothesis.

By the same argument as the one already presented, a path $\xi$
starting in $I$ labeled with $xa$ and ending in $p_i$ must reach a
state of $\Set(s')$ after reading $x$. In view of that, \eqref{eq:3}
implies that
\begin{equation}
\wI[\pi] \otimes w_i = \WI(xa, p_i),
\end{equation}
which concludes the proof.\qed
\end{proof}

\begin{figure}[t]
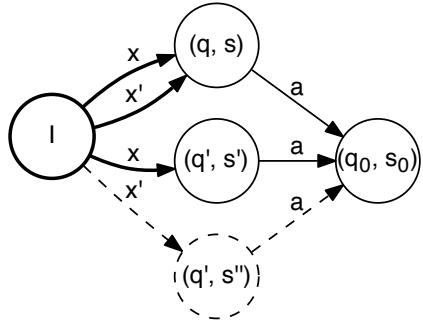

\centering
\vspace{-1cm}
\ipsfig{.66}{figure=t5} 
\caption{Illustration of the proof of Lemma~\ref{lemma:1}. The lemma
  proves the existence of the dashed path and state for $(q, s) \neq
  (q', s')$ and $x \neq x'$.} 
\label{fig:t5}
\end{figure}

\section{Proof of Lemma~\ref{lemma:1}}

\begin{proof}
  First, note that since $s = s(q, x)$ and $s' = s(q', x)$, $q = q'$
  implies $(q, s) = (q', s')$. By contraposition, since $(q, s) \neq
  (q', s')$, we must have $q \neq q'$.  Since both $q_0 \in
  \delta(q, a)$ and $q_0 \in \delta(q', a)$ in $A$ (or both $q$ and
  $q'$ are final states), $q$ and $q'$ share a common future, which
  implies $q \, \R \, q'$. Since $(q', s')$ is reachable by $x$ in
  $A'$ from $I'$, $q'$ must be reachable by $x$ from $I$ in $A$. This,
  combined with $q \, \R \, q'$, implies that $q'$ must be in
  $\Set(s)$. Since $(q, s) \in \delta(I', x')$, all states in
  $\Set(s)$ must be reachable by $x'$ from $I$ in $A$, in particular
  $q'$. Thus, by definition of the \predis\ construction, $A'$ admits
  a state $(q', s(q', x'))$, which is distinct from $(q, s)$ since $q
  \neq q'$. If $(q, s)$ admits a transition labeled with $a$ to $(q_0,
  s_0)$, then we have $s_0 = s(q_0, x'a)$.  If $(q', s')$ also admits
  a transition labeled with $a$ to $(q_0, s_0)$, then $q'$ admits a
  transition labeled with $a$ to $q_0$ and by definition of the
  \predis\ construction, $(q', s(q', x'))$ must admit a transition by
  $a$ to $(q_0, s(q_0, x'a)) = (q_0, s_0)$. Finally, in the case where
  both $(q, s)$ and $(q', s')$ are final states, then $q'$ is final in
  $A$ and thus $(q', s(q', x'))$ is a final state in $A'$.\qed
\end{proof}

\section{Proof of Lemma~\ref{lemma:11}}

\begin{proof}
  Fix $a \in \Sigma$ and let $\l(q_0, s_0, a) = ((q_1, s_1), \ldots,
  (q_n, s_n))$, $n \geq 1$, be the list of all distinct states of $A'$
  admitting a transition labeled with $a \in \Sigma$ to $(q_0, s_0)$,
  with $q_1 \leq \cdots \leq q_n$. By definition, the $a$-transition
  of the first state $(q_1, s_1)$ is kept, thus the set of strings
  accepted is unchanged after processing the first state. Assume now
  that the set of strings accepted is the same as that of $A'$ after
  processing all states $(q_1, s_1), \ldots, (q_i, s_i)$, $i \in [1, n
  - 1]$.  Assume that after processing $(q_{i + 1}, s_{i + 1})$ its
  $a$-transition to $(q_0, s_0)$ is removed, otherwise the set of strings
  accepted is clearly unchanged and is thus the same as $A'$ by the
  induction hypothesis. The removal occurs because $(q_{i + 1}, s_{i +
    1})$ and some state $(q_j, s_j)$ are both in $\delta(I', x)$ for
  some $x \in \Sigma^*$, with $j < i + 1$.  The removal of the
  transition could potentially cause the elimination of a string
  accepted by the automaton because $(q_{i + 1}, s_{i + 1})$ may be
  reachable by some other string $x' \neq x$ that does not reach
  $(q_j, s_j)$. Assume that $(q_{i + 1}, s_{i + 1})$ is reachable by
  such a string $x' \neq x$. We will show that at least one previously
  processed state is reachable by $x'$ whose  $a$-transition to
  $(q_0, s_0)$ has not been removed. This will prove that the set of
  strings accepted is not affected by the processing of $(q_{i + 1},
  s_{i + 1})$.

  Assume that no such previously processed state exists.  By
  Lemma~\ref{lemma:1}, there exists a state $(q_{k_1}, s_{k_1})$ in
  $\l(q_0, s_0, a)$ reachable by $x'$, distinct from $(q_{i + 1}, s_{i
    + 1})$ and with $q_{k_1} = q_j$. State $(q_{k_1}, s_{k_1})$ must
  have been processed before $(q_{i + 1}, s_{i + 1})$, otherwise, $j <
  i + 1 < k_1$ and $q_{k_1} = q_j$ would imply $q_j = q_{i + 1}$,
  which cannot be since, by construction, two distinct states of $A'$
  of the form $(q_j, s_j)$ and $(q_j, s_{i + 1})$ cannot be
  co-reachable.  Thus, since by assumption no previously processed
  state admitting a $a$-transition to $(q_0, s_0)$ is reachable by
  $x'$, the $a$-transition from $(q_{k_1}, s_{k_1})$ to $(q_0, s_0)$
  must have been removed. By the same assumption, the removal of the
  $a$-transition from $(q_{k_1}, s_{k_1})$ must be because $(q_{k_1},
  s_{k_1}) \in \delta(I', x'')$ and $(q_l, s_l) \in \delta(I', x'')$
  for some string $x'' \neq x'$ and some $l < k_1$, and because the
  $a$-transition of $(q_l, s_l)$ to $(q_0, s_0)$ has not been
  removed. By Lemma~\ref{lemma:1}, this implies the existence of a
  state $(q_{k_2}, s_{k_2})$ in $\l(q_0, s_0, a)$ reachable by $x'$,
  with $q_{k_2} = q_l$ and and with $(q_{k_2}, s_{k_2})$ distinct from
  $(q_{k_1}, s_{k_1})$. As argued before, this implies that $(q_{k_2},
  s_{k_2})$ has been processed before $(q_{k_1}, s_{k_1})$, therefore
  we have $k_2 < k_1$. Since $(q_{k_2}, s_{k_2})$ is reachable by
  $x'$, by assumption, its $a$-transition to $(q_0, s_0)$ must have
  been removed. Proceeding in this way, we can construct an infinite
  sequence of strictly decreasing indices $k_1 > k_2 > \ldots > k_m >
  \cdots$ of states $(q_{k_m}, s_{k_m})$ in $\l(q_0, s_0, a)$
  reachable by $x'$, which would contradict the finiteness of $\l(q_0,
  s_0, a)$.  Thus, there exists a previously processed state in
  $\l(q_0, s_0, a)$ whose $a$-transition has not been removed and that
  is reachable by $x'$, which concludes the proof of the first
  claim. The proof for processing $\f$ follows the same
  steps.\footnote{We can also introduce a \emph{super-final state} $f$
    to which all final states of $A'$ are connected by a transition
    labeled with an auxiliary symbol $\phi \not \in \Sigma$ with the
    semantics of finality. The proof is then syntactically the same as
    for regular symbols.}\qed
\end{proof}

\section{Proof of Theorem~\ref{th:weak_twins}}

\begin{proof}
  Assume that $A$ admits the $\R$-\emph{weak twins property} and that
  the \predis\ construction creates infinitely many distinct states
  $(q, s)$. Since the states defining the weighted subsets $s$ are of
  a finite number, there must be infinitely many states $(q_n, s_n)$,
  $n \!\in\! \Nset$, with the same $\sset(s_n)$. Among these states,
  we must have $q_n = q$ for at least one state $q$ for infinitely
  many $n \in \Nset$, since the number of distinct states $q_n$ is
  finite. Thus, the assumption made implies that the \predis\ construction
  creates an infinite sequence $(q, s_n)$, $n \!\in\!  \Nset$, with
  the same $\sset(s_n) = \set{p_1, \ldots, p_k}$, $k < +\infty$, and
  say $p_1 = q$. Thus, we can write $s_n = \set{(p_1, w_n(p_1)),
    \ldots, (p_k, w_n(p_k))}$.

  By Proposition~\ref{prop:1}, for any $n \in \Nset$, there exists a
  string $x_n \!\in\! \Sigma^*$, with
  \begin{equation}
    \label{eq:weights}
    \forall p \in \set{p_1, \ldots, p_k}, \quad w_n(p) = \WI(x_n, p) - \WI(x_n, \set{p_1, \ldots, p_k}).
  \end{equation}
  There exists at least one $p \in \set{p_1, \ldots, p_k}$ such that
  $\WI(x_n, \set{p_1, \ldots, p_k}) = \WI(x_n, p)$ for infinitely
  many indices $J \subseteq \Nset$, since $k$ is finite. By
  \eqref{eq:weights}, $w_n(p) = 0$ for all $n \in J$. $\set{w_n(q) -
    w_n(p_i) \colon n \in J}$ cannot be finite for all $i \in [1, k]$,
  otherwise in particular $\set{w_n(q) - w_n(p) \colon n \in J} =
  \set{w_n(q) \colon n \in J}$ would be finite, which in turn, by the
  finiteness of $\set{w_n(q) - w_n(p_i) \colon n \in J}$ for all $i$,
  would imply the finiteness $\set{w_n(p_i) \colon n \in J}$ for all $i$,
  contradicting the infiniteness of $\set{s_n\colon n \in J}$. Thus,
  there must exist at least one state $r \in \set{p_1, \ldots, p_k}$
  such that $\set{w_n(q) - w_n(r) \colon n \in J}$ is infinite.

We will show that $\set{w_n(q) - w_n(r) \colon n \in J}$ is
included in the finite set
\begin{equation}
A = \set{\wI[\pi_1] - \wI[\pi_0]\colon \pi_1 \in P(I, x, q), \pi_0 \in P(I, x,
  r), |x| \leq |Q|^2 - 1},
\end{equation}
thereby contradicting the original assumption about \predis\ creating
infinitely many states.

Refer to a shortest path with an origin at an initial state and that
includes the initial state's weight as an \emph{{\cal
    I}-shortest-path}.  Consider $x = x_n$ for some $n \in \Nset$.
Let $\pi_1$ be an {\cal I}-shortest-path among $P(I, x, q)$ and $\pi_0$ an
{\cal I}-shortest-path among $P(I, x, r)$. Thus, by \eqref{eq:weights}, we
can write
\begin{equation}
  w_n = (\wI[\pi_1] - \WI(x, p)) -  (\wI[\pi_0] - \WI(x, p)) = \wI[\pi_1] -  \wI[\pi_0].
\end{equation}
Since both $q$ and $r$ are reachable from $I$ by a path labeled with
$x$, there is a path in $A \cap A$ from a pair of initial states to
$(q, r)$. Assume that $|x| > |Q|^2 - 1$, then this path must go
through at least one non-empty cycle at some state $(q_1, r_1)$. Thus,
by definition of intersection, paths $\pi_0$ and $\pi_1$ can be
decomposed as
\begin{align*}
\pi_1 = \pi^1_1 \pi^2_1 \pi^3_1 & \quad \text{with } \pi^1_1 \in P(I, x^1,
q_1), \pi^2_1 \in P(q_1, x^2, q_1), \pi^2_1 \in P(q_1, x^3, q)\\
\pi_0 = \pi^1_0 \pi^2_0 \pi^3_0 & \quad \text{with } \pi^1_0 \in P(I, x^1,
r_1), \pi^2_0 \in P(r_1, x^2, r_1), \pi^2_0 \in P(r_1, x^3, r).
\end{align*}
Since $\pi_0$ and $\pi_1$ are shortest paths, the cycles at $q_1$ and
$r_1$ are also shortest paths.  Now, by definition of the states
created by \predis\ all states in $\set{p_1, \ldots, p_k}$, in
particular $r$, are in
$\R$-relation with $q$. By compatibility with the inverse transition
function, this implies that $r_1$ and $q_1$ are also in $\R$-relation.
Thus, by the
$\R$-weak twins property, the weight of the cycle at $q_1$ and that of the
cycle at state $r_1$ in the decompositions above must be equal.
Therefore, we can write
\begin{equation}
  w_n = \wI[\pi'_1] -  \wI[\pi'_0].
\end{equation}
with $\pi'_1 = \pi^1_1 \pi^3_1$ and $\pi'_0 = \pi^1_0 \pi^3_0$. We
have $|\pi'_1| < |\pi_1|$ and $|\pi'_0| < |\pi_0|$. Thus, by induction
on $|x|$, we can find two paths $\pi''_1 \in P(I, x'', q)$ and 
$\pi''_0 \in P(I, x'', r)$ with $w_n = \wI[\pi''_1] -  \wI[\pi''_0]$ and $|x''| \leq |Q|^2 -
1$. Proceeding in the same way for all $x_n$, this shows that
$\set{w_n\colon n \in J}$ is included in the finite
set $A$, which contradicts the fact the number of states
created by the \predis\ construction is infinite.\qed
\end{proof}

\section{Proof of Theorem~\ref{th:det_predis}}

\begin{proof}
  Let $A$ be a determinizable WFA. Assume that the
  application of \predis\ to $A$ generates an infinite set of distinct
  states. Then, as in the proof of Theorem~\ref{th:weak_twins}, this
  implies the existence of two states $q$ and $r$ reachable from the
  initial states by strings $(x_n)_{n \in \Nset}$, and such that the
  set $\set{w_n(q) - w_n(r) \colon n \in \Nset} = \set{ \WI(x_n, q) -
    \WI(x_n, r) \colon n \in \Nset}$ is infinite.

  For any $n \in \Nset$, consider the weighted subset $S_n$
  constructed by weighted determinization which is the set of pairs
  $(p, v)$, where $p$ is a state of $A$ reachable by $x_n$ from the
  initial state and $v$ its \emph{residual weight} defined by $v =
  \WI(x_n, p) - \min_{p' \in \delta(I, x_n)} \WI(x_n, p')$. 
  Since $A$ is determinizable, there can only be finitely many
  distinct $S_n$, $n \in \Nset$.  $S_n$ includes the pairs $(q, v_n)$
  and $(r, v'_n)$ with $v_n = \WI(x_n, q) - v_0$ and $v'_n = 
  \WI(x_n, q) - v_0$, where $v_0$ is the weight of an {\cal I}-shortest-path labeled
  with $x_n$ and starting at the initial states. Since the number of
  distinct weighted subsets $S_n$ is finite, so must be the number of
  distinct pairs $((q, v_n), (r, v'_n))$ they each include. This
  implies that there are only finitely many distinct differences of
  weight in $\set{v'_n - v_n \colon n \in I}$.  But, since $v'_n - v_n
  = \WI(x_n, q) - \WI(x_n, r)$, this contradicts the infiniteness of
  $\set{ \WI(x_n, q) - \WI(x_n, r) \colon n \in \Nset}$. Thus,
  \predis\ cannot generate an infinite number of states and $A$ is
  \predisable.\qed
\end{proof}

\end{document}